\documentclass[11pt]{article}

\usepackage{fullpage}
\usepackage{amssymb, amsmath, amsthm}
\usepackage{authblk}
\usepackage{algorithm, algpseudocode, eqparbox, thmtools, thm-restate, enumerate, verbatim, bm, array, tabu}
\usepackage[usenames, dvipsnames, svgnames, table]{xcolor}
\usepackage[colorlinks=true,pdfpagemode=UseNone,
citecolor=OliveGreen,linkcolor=BrickRed,urlcolor=BrickRed,
pdfstartview=FitH,pagebackref=true]{hyperref}
\usepackage{mathtools}

\newcommand{\Input}{\item[{\bf Input:}]}
\newcommand{\Output}{\item[{\bf Output:}]}

\renewcommand{\P}[1]{{\mathbb{P}}\left[#1\right]}

\newcommand{\PP}[2]{{\mathbb{P}}_{#1}\left[#2\right]}

\newcommand{\EE}[2]{{\mathbb{E}}_{#1}\left[#2\right]}

\providecommand{\norm}[1]{\left\lVert#1\right\rVert}

\def\equationautorefname~#1\null{
  equation~(#1)\null
}

\declaretheorem[numberwithin=section]{theorem}
\declaretheorem[sibling=theorem]{lemma}
\declaretheorem[sibling=theorem]{proposition}
\declaretheorem[sibling=theorem]{claim}
\declaretheorem[sibling=theorem]{conjecture}
\declaretheorem[sibling=theorem]{corollary}

\declaretheorem[sibling=theorem]{fact}
\declaretheorem[sibling=theorem]{definition}

\declaretheorem[sibling=theorem]{problem}

\def\setminus{-}
\def\br{\delta}

\newenvironment{proofof}[1]{{\medbreak\noindent \em Proof of #1.  }}{\hfill\qed\medbreak}

\def\bone{{\bf 1}}
\def\bzero{{\bf 0}}
\def\eps{{\epsilon}}
\def\R{\mathbb{R}}

\def\bx{{\bf x}}

\def\cF{{\cal F}}

\def\blambda{{{\bm{\lambda}}}}
\def\bmu{{{\bm{\mu}}}}
\def\bnu{{{\bm{\nu}}}}

\def\bx{{{\bf{x}}}}

\def\bv{{{\bf{v}}}}
\def\br{{{\bf{r}}}}

\def\bz{{{\bf{z}}}}

\def\l{\ell}

\DeclareMathOperator{\girth}{girth}

\DeclareMathOperator{\poly}{poly}

\DeclareMathOperator{\polyloglog}{polyloglog}

\DeclareMathOperator{\argmax}{argmax}
\DeclareMathOperator{\argmin}{argmin}
\DeclareMathOperator{\trace}{Tr}

\DeclareMathOperator{\avg}{avg}
\newcommand{\maxdeg}{\deg_{\max}}
\newcommand{\avgdeg}{\deg_{\avg}}

\title{Approximating the Largest Root and\\ Applications to Interlacing Families}

\author[1]{Nima Anari}
\author[2]{Shayan Oveis Gharan}
\author[1]{Amin Saberi}
\author[3]{Nikhil Srivastava}

\affil[1]{\small Stanford University\\ \texttt{\{anari,saberi\}@stanford.edu}}
\affil[2]{\small University of Washington\\ \texttt{shayan@cs.washington.edu}}
\affil[3]{\small University of California, Berkeley\\ \texttt{nikhil@math.berkeley.edu}}

\date{\today}

\begin{document}
\maketitle

\begin{abstract}
We study the problem of  approximating the largest root of a real-rooted polynomial of degree $n$ using its top $k$ coefficients and give nearly matching upper and lower bounds. We present algorithms with running time polynomial in $k$ that use the top $k$ coefficients to approximate the maximum root within a factor of $n^{1/k}$ and $1+O(\tfrac{\log n}{k})^2$ when $k\leq \log n$ and $k>\log n$ respectively. 
        We also prove   corresponding information-theoretic lower bounds of 	$n^{\Omega(1/k)}$ and $1+\Omega\left(\frac{\log \frac{2n}{k}}{k}\right)^2$, and show strong lower bounds for noisy version of the problem in which one is given access to approximate coefficients.

This problem has applications in the context of the method of interlacing families of polynomials, which was used for proving the existence of Ramanujan graphs of all degrees, the solution of the Kadison-Singer
problem, and bounding the integrality gap of the asymmetric traveling salesman
problem. All of these involve computing the maximum root of certain real-rooted
polynomials for which the top few coefficients are accessible in subexponential
time. Our results yield an algorithm with the running time of $2^{\tilde
O(\sqrt[3]n)}$ for all of them.
\end{abstract}
\tableofcontents
\section{Introduction}

For a non-negative vector $\bmu=(\mu_1,\dots,\mu_n)\in\R_+^n$, let $\chi_{\bmu}$ denote the unique monic polynomial with roots $\mu_1,\dots,\mu_n$:
\[ \chi_{\bmu}(x):=\prod_{i=1}^n(x-\mu_i). \]

Suppose that we do not know $\bmu$, but rather know the top $k$ coefficients of $\chi_{\bmu}$ where $1\leq k< n$. In more concrete terms, suppose that
\[ \chi_{\bmu}(x)=x^n+c_1x^{n-1}+c_2x^{n-2}+\dots+c_n, \]
and we only know $c_1,\dots, c_k$. What do $c_1,\dots, c_k$ tell us about the roots and in particular $\max_{i}\mu_{i}$? 

\begin{problem}
	\label{prob:main}
Given the top $k$ coefficients of a real rooted polynomial of degree $n$, how well can you approximate its largest root?
\end{problem}

This problem may seem completely impossible if $k$ is significantly smaller than $n$.
For example, consider the two polynomials $x^n-a$ and $x^n-b$. The largest roots of these two polynomials can differ arbitrarily in absolute value and even by knowing the top $n-1$ coefficients we can not approximate the absolute value of the largest at all. 

The key to approach the above problem is to exploit real rootedness. One approach is to construct a polynomial with the given coefficients and study its roots.
Unfortunately, even assuming that the roots of the original polynomial are real and well-separated, adding an exponentially small amount of noise to the (bottom) coefficients can lead to constant sized perturbations of the roots in the complex plane --- the most
famous example is Wilkinson's polynomial \cite{Wilkinson84}: $$(x-1)(x-2)\ldots
(x-20).$$ 

 Instead, we use the given coefficients to compute a polynomial of the roots, e.g., the $k$-th moment of the roots, and we use that polynomial to estimate the largest root.
Our contributions towards \autoref{prob:main} are as follows. 

\paragraph{Efficient Algorithm for Approximating the Largest Root.} 
In \autoref{thm:positive}, we present an upper bound showing that we can use the top $k$ coefficients of a real rooted polynomial with nonnegative roots to efficiently obtain an $\alpha_{k,n}$ approximation of the largest root where
    \[
    	\alpha_{k,n}=\begin{cases}
        	n^{1/k}&k\leq \log n,\\
            1+O(\frac{\log n}{k})^2&k>\log n.
        \end{cases}
    \]
    Moreover such an approximation can be done in $\poly(k)$ time. This implies that exact access to  $O\left(\frac{\log n}{\sqrt{\epsilon}}\right)$
		coefficients is sufficient for
		$(1+\epsilon)$-approximation of the largest root. 

\paragraph{Nearly Matching Lower Bounds.} 
The main nontrivial part of this work is our information-theoretic matching lower bounds. 
In \autoref{thm:negative}, we show that 
when $k<n^{1-\epsilon}$ and for some constant $c$, there are no algorithms with approximation factor better than $\alpha_{k,n}^{c}$. 
Chebyshev polynomials are critical for this construction as well. We also use known constructions for large-girth graphs, and the proof of \cite{MSS12} for a variant of Bilu-Linial's conjecture \cite{BL06}.
    Our bounds can be made slightly sharper assuming Erd\H{o}s's girth conjecture.


 \subsection{Motivation and Applications}
For many important polynomials, it is easy to compute the top $k$ coefficients exactly, whereas it is provably hard to compute all of them. One example is the matching polynomial of a graph, whose coefficients encode the number of matchings of various sizes. For this polynomial, computing the constant term, i.e. the number of perfect matchings, is \#P-hard \cite{Valiant79}, whereas for small $k$, one can compute the number of matchings of size $k$ exactly, in time $n^{O(k)}$, by simply enumerating all possibilities. Roots of the matching polynomial, and in particular the largest root arise in a number of important applications \cite{HL72, MSS12, SS13}. So it is natural to ask how well the largest root can be approximated from the top few coefficients. 

Another example of a polynomial whose top coefficients are easy to compute is the independence polynomial of a graph, which is real rooted for claw-free graphs \cite{CS07}, and whose roots have connections to the Lov\'{a}sz Local Lemma (see e.g. \cite{HSV16}). 

\paragraph{Subexponential Time Algorithms for Method of Interlacing Polynomials}
Our main motivation for this work is the method of interlacing families of
polynomials~\cite{MSS12, MSS13, AO15}, which has been an essential tool in the
development of several recent results including the construction of Ramanujan
graphs via lifts \cite{MSS12, MSS15, HPS15}, the solution of the Kadison-Singer
problem \cite{MSS13}, and improved integrality gaps for the asymmetric traveling
salesman problem \cite{AO14}. Unfortunately, all these results on the existence
of expanders, matrix pavings, and thin trees, have the drawback of being
nonconstructive, in the sense that they do not give polynomial time algorithms
for finding the desired objects (with the notable exception of \cite{Cohen16}).
As such, the situation is somewhat similar to that of the Lov\'{a}sz Local Lemma
(which is able to guarantee the existence of certain rare objects
nonconstructively), before algorithmic proofs of it were found \cite{Beck91,
MT10}.


In \autoref{section:apps}, we use \autoref{thm:positive} to give a $2^{\tilde
O(\sqrt[3]{m})}$ time algorithm for rounding an interlacing family of depth $m$,
improving on the previously known running time of $2^{O(m)}$. This leads to
algorithms of the same running time for all of the problems mentioned above.

\paragraph{Lower Bounds Given Approximate Coefficients}
In the context of efficient algorithms, one might imagine that computing a much larger number of coefficients {\em approximately} might provide a better estimate of the largest root. In particular, we consider the following noisy version of \autoref{prob:main}:
\begin{problem}\label{prob:approx}
	Given real numbers $a_1,\ldots,a_k$, promised to be
	$(1+\delta)-$approximations of the first $k$ coefficients of a
	real-rooted polynomial $p$, how well can you approximate the largest
	root of $p$?
\end{problem}

An important extension of our information theoretic lower bounds is that
\autoref{prob:main} is extremely sensitive to noise: in
\autoref{prop:approxlb} we prove that even knowing  {\em all} but the $k$-th
coefficient exactly and knowing the $k$-th one up to a $1+1/2^k$ error is no
better than knowing only the first $k-1$ coefficients exactly. We do this by
exhibiting two polynomials which agree on all their coefficients except for the
$k^{th}$, in which they differ slightly, but nonetheless have very different
largest roots. 

This example is relevant in the context of interlacing families, because the
polynomials in our lower bound have a common interlacing and are characteristic polynomials of $2-$lifts of a base
graph, which means they could actually arise in the
proofs of \cite{MSS12, MSS13}. To appreciate this more broadly, one can consider the following taxonomy of increasingly structured
polynomials:
\begin{align*}\textrm{complex polynomials} &\supset \textrm{real-rooted polynomials} \supset
\textrm{mixed characteristic polynomials}\\
&\supset \textrm{characteristic polynomials of lifts of graphs}.\end{align*}
Our example complements the standard numerical analysis wisdom (as in
Wilkinson's example) that complex polynomial roots are in general terribly
ill-conditioned as functions of their coefficients, and shows that this fact
remains true even in the structured setting of interlacing families.

\autoref{prop:approxlb} is relevant to the quest for efficient algorithms for interlacing families
for the following reason. All of the coefficients of the matching polynomial of
a bipartite graph can be approximated to $1+1/\poly(n)$ error in polynomial
time, for any fixed polynomial, using Markov Chain Monte Carlo techniques
\cite{JS89,JSV04,friedland}. One might imagine that an extension of these
techniques could be used to approximate the coefficients of the more general
expected characteristic polynomials that appear in applications of interlacing
families. In fact for some families of interlacing polynomials (namely, the
mixed characteristic polynomials of \cite{MSS13}) we can design
Markov chain Monte Carlo techniques to approximate the top half of the
coefficients within $1+1/\poly(n)$ error. 

Our information theoretic lower bounds  rule out this method as a way to
approximate the largest root, at least in full generality, since knowing  all of
the coefficients of a real-rooted polynomial up to a $(1+1/\poly(n))$ error for any $\poly(n)$ is no
better than just knowing the first $\log n$ coefficients exactly, in the worst
case, even under the promise that the given polynomials have a common
interlacing.  In other words, even an MCMC oracle that gives $1+1/\poly(n)$
approximation of all coefficients would not generically allow one to
round an interlacing family of depth greater than logarithmic in $n$, since the
error accumulated at each step would be $1/\mathrm{polylog}(n)$.

\paragraph{Connections to Poisson Binomial Distributions} Finally, there is a
probabilistic view of \autoref{prob:main}. Assume that $X=B(p_1)+\dots+B(p_n)$
is a sum of independent Bernoulli random variables, i.e. a Poisson binomial,
with parameters $p_1,\dots,p_n\in[0,1]$. Then \autoref{prob:main} becomes the
following: Given the first $k$ moments of $X$ how well can we approximate
$\max_i p_i$? In this view, our paper is related to \cite{CD15}, where it was
shown that any pair of such Poisson binomial random variables with the same
first $k$ moments have total variation distance at most $2^{-\Omega(k)}$.
However, the bound on the total variation distance does not directly imply a
bound on the maximum $p_i$. 

\paragraph{Discussion}
Besides conducting a precise study of the dependence of the largest root of a
real-rooted polynomial on its coefficients, the results of this
paper shed light on what a truly efficient algorithm for interlacing families
might look like. On one hand, our running time of $2^{\tilde O(m^{1/3})}$ shows
that the problem is not ETH hard, and is unnatural enough to suggest that a faster algorithm (for instance,
quasipolynomial) may exist. On the other hand, our lower bounds show
that the polynomials that arise in this method are in general hard to compute in a rather robust sense: namely, obtaining an inverse
polynomial error approximation of their largest roots requires knowing {\em many} coefficients {\em
exactly}. This implies that in order to obtain an efficient algorithm for
even approximately simulating the interlacing families proof technique, one will have to exploit finer
properties of the polynomials at hand, or
find a more ``global'' proof which is able to reason about the error in
a more sophisticated amortized manner, or perhaps track a more well-conditioned
quantity in place of the largest root, which can be computed using fewer
coefficients and which still satisfies an approximate interlacing property. 

\section{Preliminaries} \label{preliminaries}

We let $[n]$ denote the set $\{1,\dots,n\}$. We use the notation
$\binom{[n]}{k}$ to denote the family of subsets $T\subseteq[n]$ with $|T|=k$.
We let $S_n$ denote the set of permutations on $[n]$, i.e. the set of bijections
$\sigma:[n]\to[n]$.

We use bold letters to denote vectors. For a vector $\bmu\in\R^n$, we denote its
coordinates by $\mu_1,\dots,\mu_n$. We let $\mu_{\max}$ and $\mu_{\min}$ denote
$\max_i \mu_i$ and $\min_i \mu_i$ respectively.

For a symmetric matrix $A$, we denote the vector of eigenvalues of $A$, i.e. the
roots of $\det(xI-A)$, by $\blambda(A)$. Similarly we denote the largest and
smallest eigenvalues by $\lambda_{\max}(A)$ and $\lambda_{\min}(A)$. We slightly
abuse notation, and for a polynomial $p$ we write $\blambda(p)$ to denote the
vector of roots of $p$. We also write $\lambda_{\max}(p)$ to denote the largest
root of $p$. 

For a graph $G=(V,E)$ we let $\maxdeg(G)$ denote the maximum degree of its
vertices and $\avgdeg(G)$ denote the average degree of its vertices, i.e.
$2|E|/|V|$.

What follows are mostly standard facts; the proofs of
\autoref{fact:linear-transform}, \autoref{fact:cheb-large-x},
\autoref{fact:cheb-large-x}, and \autoref{fact:avgdeg} are included in
\autoref{appendix:prelims} for completness.

\paragraph{Facts from Linear Algebra}
For a matrix $A\in\R^{n\times n}$, the characteristic polynomial of $A$ is defined as
$ \det(xI-A)$.
Letting $\sigma_k(A)$ be the sum of all principal $k$-by-$k$ minors of $A$, we
have: 
$$ \det(xI-A)=\sum_{k=0}^n x^{n-k} \sigma_k(A).$$
There are several  algorithms that for a matrix $A\in \R^{n\times n}$ calculate $\det(xI-A)$ in time polynomial in $n$. By the above identity, we can use any such algorithm to efficiently obtain $\sigma_k(A)$ for any $1\leq k\leq n$.

The following proposition is proved in \cite{MSS13} using the Cauchy-Binet formula.
\begin{proposition}
\label{prop:sigmak}
Let $\bv_1,\dots,\bv_m\in \R^n$. Then,
$$\det\left(xI-\sum_{i=1}^m \bv_i\bv_i^T\right) = \sum_{k=0}^n (-1)^k x^{n-k}\sum_{S\subseteq \binom{[m]}{k}} \sigma_k\left(\sum_{i\in S} \bv_i\bv_i^T\right).$$
\end{proposition}

\paragraph{Symmetric Polynomials}


We will make heavy use of the elementary symmetric polynomials, which relate the
roots of a polynomial to its coefficients.
\begin{definition}
	Let $e_k\in \R[\mu_1,\dots,\mu_n]$ denote the $k$-th elementary symmetric polynomial defined as
\[ e_k(\bmu):=\sum_{T\in \binom{[n]}{k}}\prod_{i\in T}\mu_i. \]
\end{definition}

\begin{fact}
	\label{fact:ekck}
	Consider the monic univariate polynomial $\chi(x)=x^n+c_1x^{n-1}+\dots+c_n$.	 Suppose that $\mu_1,\dots,\mu_n$ are the roots of $\chi$. Then for every $k\in[n]$,
	\[ c_k=(-1)^k e_k(\mu_1,\dots,\mu_n). \]
\end{fact}
This means that knowing the top $k$ coefficients of a polynomial is equivalent to knowing the first $k$ elementary symmetric polynomials of the roots.
It also implies the following fact about how shifting and scaling affect the elementary symmetric polynomials.
\begin{fact}
	\label{fact:linear-transform}
	Let $\bmu, \bnu\in \R^n$ be such that $e_i(\bmu)=e_i(\bnu)$ for $i=1,\dots, k$. If $a,b\in\R$ then $	e_i(a\bmu+b)=e_i(a\bnu+b)$ for $i=1,\dots,k$.
\end{fact}

We will use the following relationship between the elementary symmetric
polynomials and the power sum polynomials.

\begin{theorem}[Newton's Identities]
	\label{thm:newton}
	For $1\leq k\leq n$, the polynomial
	$p_k(\bmu):=\sum_{i=1}^n \mu_i^k$
	can be written as
	$q_k(e_1(\bmu),\dots,e_k(\bmu)),$
	where $q_k\in\R[e_1,\dots, e_k]$. Furthermore, $q_k$ can be computed at any point in time $\poly(k)$.
\end{theorem}

One of the immediate corollaries of the above is the following.

\begin{corollary}
	\label{cor:compute}
	Let $p(x)\in\R[x]$ be a univariate polynomial with $\deg p\leq k$. Then $\sum_{i=1}^n p(\mu_i)$ can be written as
	\[ q(e_1(\bmu),\dots,e_k(\bmu)), \]
	where $q\in\R[e_1,\dots,e_k]$. Furthermore, $q$ can be computed at any point in time $\poly(k)$.
\end{corollary}

\autoref{thm:newton} shows how $p_1,\dots,p_k$ can be computed from $e_1,\dots,e_k$. The reverse is also true. A second set of identities, also known as Newton's identities, imply the following.

\begin{theorem}[Newton's Identities]
	For each $k\in[n]$, $e_k(\bmu)$ can be written as a polynomial of $p_1(\bmu),\dots,p_k(\bmu)$ which can be computed in time $\poly(k)$.
\end{theorem}
A corollary of the above and \autoref{thm:newton} is the following.
\begin{corollary}
	\label{cor:ekpk}
	For two vectors $\bmu,\bnu\in\R^n$, we have
	\[ \left(\forall i\in[k]: e_i(\bmu)=e_i(\bnu)\right)  \iff \left(\forall i\in[k]: p_i(\bmu)=p_i(\bnu)\right). \]
\end{corollary}

\paragraph{Chebyshev Polynomials}

Chebyshev polynomials of the first kind, which we will simply call Chebyshev polynomials, are defined as follows.
\begin{definition}\label{def:chebyshevpolyn}
	Let the polynomials $T_0,T_1,\dots\in\R[x]$ be recursively defined as
	\begin{align*}
		T_0(x)&:=1,\\
		T_1(x)&:=x,\\
		T_{n+1}(x)&:=2xT_n(x)-T_{n-1}(x).
	\end{align*}
	We will call $T_k$ the $k$-th Chebyshev polynomial.
\end{definition}
Notice that the coefficients of $T_k$ can be computed in $\poly(k)$ time, by the above recurrence for example. Chebyshev polynomials have many useful properties, some of which we mention below. For further information, see \cite{Szeg39}.

\begin{fact}
	For $k\geq 0$ and $\theta\in\R$, we have
	\begin{align*}
		 T_k(\cos(\theta))&=\cos(k\theta),\\
		 T_k(\cosh(\theta))&=\cosh(k\theta).
	\end{align*}
\end{fact}

\begin{fact}
	The $k$-th Chebyshev polynomial $T_k$ has degree $k$.
\end{fact}

\begin{fact}
	\label{fact:cheb-small-x}
	For any $x\in[-1,1]$, we have $T_k(x)\in[-1,1]$.
\end{fact}

\begin{fact}
	\label{fact:cheb-large-x}
For any integer $k\geq 0$,	$T_k(1+x)$ is monotonically increasing for $x\geq 0$. Furthermore for $x\geq 0$,
	\[ T_k(1+x)\geq (1+\sqrt{2x})^k/2. \]
\end{fact}

In our approximate lower bound we will use the following connection between
Chebyshev polynomials and graphs, due to Godsil and Gutman \cite{GG81}.

\begin{fact}\label{fact:godsil} If $A_n$ is the adjacency matrix of a cycle on $n$ vertices, then
	$$\det(2xI-A_n) = 2T_n(x).$$
\end{fact}

\paragraph{Graphs with Large Girth}

In order to prove some of our impossibility results, we use the existence of extremal graphs with no small cycles.
\begin{definition}
	For an undirected graph $G$, we denote the length of its shortest cycle by $\girth(G)$. If $G$ is a forest, then $\girth(G)=\infty$.
\end{definition}

The following conjecture by Erd\H{o}s characterizes extremal graphs with no small cycles.
\begin{conjecture}[Erd\H{o}s's girth conjecture \cite{Erd64}]
	\label{conj:erdos}
	For every integer $k\geq 1$ and sufficiently large $n$, there exist graphs $G$ on $n$ vertices with $\girth(G)>2k$ that have $\Omega(n^{1+1/k})$ edges, or in other words satisfy $\avgdeg(G)=\Omega(n^{1/k})$.
\end{conjecture}
This conjecture has been proven for $k=1,2,3,5$ \cite{Wen91}. We will use the
following more general construction of graphs of somewhat lower girth.


\begin{theorem}[\cite{LU95}]
	\label{thm:large-girth}
	If $d$ is a prime power and $t\geq 3$ is odd, there is a $d$-regular bipartite graph $G$ on $2d^t$ vertices with $\girth(G)\geq t+5$.
\end{theorem}
%
%

\paragraph{Signed Adjacency Matrices} Our lower bounds will also utilize facts
about signings of graphs.

\begin{definition}
	For a graph $G=([n],E)$, we define a signing to be any function $s:E\to \{-1,+1\}$. We define the signed adjacency matrix $A_s\in\R^{n\times n}$, associated with signing $s$, as follows
	\[
		A_s(u, v):=\begin{cases}
			0&\{u, v\}\notin E,\\
			s(\{u, v\})&\{u, v\} \in E.\\
		\end{cases}
	\]
\end{definition}
Note that by definition, $A_s$ is symmetric and has zeros on the diagonal. The following fact is immediate.
\begin{fact}
	\label{fact:As-roots}
	For a signed adjacency matrix $A_s$ of a graph $G$, the eigenvalues $\blambda(A_s)$, i.e. the roots of
	$ \chi(x):=\det(xI-A_s)$,
	are real. If $G$ is bipartite, the eigenvalues are symmetric about the origin (counting multiplicities).
\end{fact}

Signed adjacency matrices were used in \cite{MSS12} to prove the existence of bipartite Ramanujan graphs of all degrees. We state one of the main results of \cite{MSS12} below.
\begin{theorem}[\cite{MSS12}]
	For every graph $G$, there exists a signing $s$ such that \[\lambda_{\max}(A_s)\leq 2\sqrt{\maxdeg(G)-1}.\]
\end{theorem}
By \autoref{fact:As-roots}, we have the following immediate corollary.
\begin{corollary}
	\label{cor:signing}
	For every bipartite graph $G$, there exists a signing $s$ such that the eigenvalues of $A_s$ have absolute value at most $2\sqrt{\maxdeg(G)-1}$.
\end{corollary}

We note that trivially signing every edge with $+1$ is often far from achieving the above bound as witnessed by the following fact.
\begin{fact}
	\label{fact:avgdeg}
	Let $A$ be the adjacency matrix of a graph $G=([n],E)$ (i.e. the signed adjacency matrix where the sign of every edge is $+1$). Then the maximum eigenvalue of $A$ is at least $\avgdeg(G)$.
\end{fact}

\section{Approximation of the Largest Root} 

In this section we give an answer to \autoref{prob:main}. As witnessed by \autoref{fact:ekck}, knowing the top $k$ coefficients of the polynomial $\chi_{\bmu}$ is the same as knowing $e_1(\bmu),\dots,e_k(\bmu)$. Therefore, without loss of generality and more conveniently, we state the results in terms of knowing $e_1(\bmu),\dots,e_k(\bmu)$.

\begin{theorem}
	\label{thm:positive}
	There is an algorithm that receives $n$ and $e_1(\bmu),\dots,e_k(\bmu)$ for some unknown $\bmu\in\R_+^n$ as input and outputs $\mu_{\max}^*$, an approximation of $\mu_{\max}$, with the guarantee that
	\[ \mu_{\max}^*\leq \mu_{\max}\leq \alpha_{k,n}\cdot \mu_{\max}^*, \]
	where the approximation factor $\alpha_{k,n}$ is
	\[
		\alpha_{k,n} = \begin{cases}
 			n^{1/k}&k\leq \log n,\\
 			1+O(\frac{\log n}{k})^2&k> \log n.
		\end{cases}
	\]
	Furthermore the algorithm runs in time $\poly(k)$.
\end{theorem}

Note that there is a change in the behavior of the approximation factor in the two regimes $k\ll \log n$ and $k\gg \log n$. When $k>\log n$, the expression $n^{1/k}$ is $1+\Theta(\frac{\log n}{k})$ which can be a much worse bound compared to $1+O(\frac{\log n}{k})^2$. When  $k$ is near the threshold of $\log n$, $n^{1/k}$ and $1+\Theta(\frac{\log n}{k})^2$ are close to each other and both of the order of $1+\Theta(1)$.

We complement this result by showing information-theoretic lower bounds.
\begin{theorem}
	\label{thm:negative}
	For every $1\leq k< n$, there are two vectors $\bmu,\bnu\in\R_+^n$ such that $e_i(\bmu)=e_i(\bnu)$ for $i=1,\dots,k$, and
	\[ \frac{\nu_{\max}}{\mu_{\max}}\geq\beta_{k,n}, \]
	where
	\[
		\beta_{k,n}= \begin{cases}
			n^{\Omega(1/k)}&k\leq \log n,\\
			1+\Omega\left(\frac{\log \frac{2n}{k}}{k}\right)^2&k> \log n.
		\end{cases}	
 	\]	
\end{theorem}
This shows that no algorithm can approximate $\mu_{\max}$ by a factor better than $\beta_{k,n}$ using $e_1(\bmu),\dots,e_k(\bmu)$. Note that for $k<n^{1-\epsilon}$,  $\beta_{k,n}=\alpha_{k,n}^c$ for some constant $c$ bounded away from zero.
For constant $k$, it is possible to give a constant multiplicative bound assuming Erd\H{o}s's girth conjecture.
\begin{theorem}
	\label{thm:cond-negative}
	Assume that $k$ is fixed and Erd\H{o}s's girth conjecture (\autoref{conj:erdos}) is true for graphs of girth $>2k$. Then for large enough $n$ there are two vectors $\bmu,\bnu\in\R_+^n$ such that $e_i(\bmu)=e_i(\bnu)$ for $i=1,\dots,k$ and
	\[ \frac{\nu_{\max}}{\mu_{\max}}\geq \Omega(n^{1/k}). \]
\end{theorem}

\subsection{Proof of Theorem \ref{thm:positive}: An Algorithm for Approximating the Largest Root}

We consider two cases: if $k\leq \log n$ we return $(p_k(\bmu)/n)^{1/k}$ as the estimate for the maximum root.
It is not hard to see that in this case, $(p_k(\bmu)/n)^{1/k}$ gives an $n^{1/k}$ approximation of the maximum root (see \autoref{claim:small-k} below).

For $k>\log n$ we can still use $(p_k(\bmu)/n)^{1/k}$ to estimate the maximum root, but this only guarantees a $1+O(\frac{\log n}{k})$ approximation.
We show that using the machinery of Chebyshev polynomials we can obtain a  better bound. 
The pseudocode for the algorithm can be see in \autoref{alg:main}. 

\begin{algorithm}
\caption{Algorithm For Approximating the Maximum Root From Top Coefficients}
\label{alg:main}
\begin{algorithmic}
\Input $n$ and $e_1(\bmu),e_2(\bmu),\dots,e_k(\bmu)$ for some $\bmu\in\R_+^n$.
\Output $\mu_{\max}^*$, an approximation of $\mu_{\max}$.
\State
\If{$k\leq \log n$}
	\State Compute $p_k(\bmu)=\sum_{i=1}^n\mu_i^k$ using Newton's identities (\autoref{thm:newton}).
	\State \Return $(p_k(\bmu)/n)^{1/k}$.
\Else
	\State $t\leftarrow e_1(\bmu)$.
	\Loop
		\State Compute $p(\bmu):=\sum_{i=1}^n T_k(\frac{\mu_i}{t})$ using \autoref{cor:compute}.
		\If{$p(\bmu)>n$}
			\State \Return $\mu_{\max}^*\leftarrow t$.
		\EndIf
		\State $t\leftarrow \frac{t}{1+(\frac{20\log n}{k})^2}$.
	\EndLoop
\EndIf
\end{algorithmic}
\end{algorithm}

We will prove the following claims to show the correctness of \autoref{alg:main}. Let us start with the case $k\leq \log n$. 
\begin{claim}
	\label{claim:small-k}
	For any $k\geq 1$ we have
	\[ \Big(\frac{p_k(\bmu)}{n}\Big)^{1/k}	\leq \mu_{\max}\leq p_k(\bmu)^{1/k}. \]
\end{claim}
\begin{proof}
	Observe,
	\[ p_k(\bmu)/n = \sum_{i=1}^n\frac{\mu_i^k}{n} \leq \sum_{i=1}^n \frac{\mu_{\max}^k}{n} = \mu_{\max}^k \leq \sum_{i=1}^n\mu_i^k =p_k(\bmu), \]
    Taking $\frac1k$-th root of all sides of the above proves the claim. 
\end{proof}
The rest of the section handles the case where $k>\log n$.
Our first claim shows that  as long as $t\geq \mu_{\max}$, the algorithm keeps decreasing $t$ by a multiplicative factor of $(1-\Omega(\log(n)/k)^2)$.
Since at the beginning we have $t=e_1(\bmu) \geq \mu_{\max}$,
we will have $\mu_{\max}^*\leq \mu_{\max}$.

\begin{claim}
	\label{claim:large-k-sound}
	For any $t\geq \mu_{\max}$,
	\[ \sum_{i=1}^n T_k(\frac{\mu_i}{t})\leq n. \]
\end{claim}
\begin{proof}
	If $t\geq \mu_{\max}$, then $\mu_i/t\in [0,1]$ for every $i\in[n]$. By \autoref{fact:cheb-small-x}, we have
	\[
		\sum_{i=1}^n T_k(\frac{\mu_i}{t})\leq \sum_{i=1}^n 1=n.
	\]
\end{proof}

To finish the proof of correctness it is enough to show that $\mu_{\max}\leq \mu_{\max}^* (1+O(\log n/k)^2)$.
This is done in the next claim. It shows that as soon as $t$ gets lower than $\mu_{\max}$, within one more iteration of the loop, the algorithm terminates.

\begin{claim}
	\label{claim:large-k-correct}
	For $k>\log n$ and $t>0$, if $\mu_{\max}>(1+(\frac{20\log n}{k})^2)t$, then
	\[ \sum_{i=1}^n T_k(\frac{\mu_i}{t})>n. \]
\end{claim}
\begin{proof}
When $\mu_{\max}>(1+(\frac{20\log n}{k})^2)t$, by \autoref{fact:cheb-large-x} we have
	\begin{align*}
		T_k(\frac{\mu_{\max}}{t}) & \geq T_k\left(1+\Big(\frac{20\log n}{k}\Big)^2\right)\geq \frac12 \left(1+\sqrt{2}\cdot \frac{20\log n}{k}\right)^k\\
		& \geq \frac12 \exp\left(\frac{3\log n}{k}\right)^k > 2n,
	\end{align*}
	where we used the inequality $1+\sqrt{800}x\geq e^{3x}$ for $x\in [0, 1]$.
	
	Now we have
	\[
		\sum_{i=1}^n T_k(\frac{\mu_i}{t})=T_k(\frac{\mu_{\max}}{t})+\sum_{i\neq \argmax_j \mu_j}T_k(\frac{\mu_i}{t})\geq 2n-(n-1)>n,
	\]
	where we used \autoref{fact:cheb-small-x} and \autoref{fact:cheb-large-x} to conclude $T_k(\frac{\mu_i}{t})\geq -1$ for every $i$.
\end{proof}
The above claim also gives us a bound on the number of iterations in which the algorithm terminates. This is because we  start the loop with $t= e_1(\bmu)\leq n\mu_{\max}$ and the loop terminates within one iteration as soon as $t<\mu_{\max}$. Therefore the number of iterations is at most
\[ 1+\frac{\log n}{\log \left(1+(\frac{20\log n}{k})^2\right)}=O\left(\log n\cdot (\frac{k}{\log n})^2\right)=O(k^2). \]







\subsection{Proofs of Theorems \ref{thm:negative} and \ref{thm:cond-negative}: Matching Lower Bounds}

The machinery of Chebyshev polynomials was used to prove \autoref{thm:positive}. We show that this machinery can also be used to prove a weaker version of \autoref{thm:negative}.
\begin{theorem}
	\label{thm:weak}
	For every $1\leq k<n$, there are $\bmu,\bnu\in\R_+^n$ such that $e_i(\bmu)=e_i(\bnu)$ for $i=1,\dots,k$ and
    \[ \frac{\nu_{\max}}{\mu_{\max}}\geq 1+\Omega(1/k^2) \]
\end{theorem}
\begin{proof}
	First let us prove this for $k=n-1$. Let $\bmu$ be the set of roots of $T_n(x-1)+1$ and $\bnu$ the set of roots of $T_n(x-1)-1$. These two polynomials are the same except for the constant term. It follows that $e_i(\bmu)=e_i(\bnu)$ for $i=1,\dots,n-1$. We use the following lemma to prove that $\bmu,\bnu\in\R_+^n$.
    \begin{lemma}
    	\label{lem:cheb-roots}
    	For $\theta\in\R$, the roots of $T_n(x)-\cos(\theta)$, counting multiplicities, are $\cos(\frac{\theta+2\pi i}{n})$ for $i=0,\dots, n-1$.
    \end{lemma}
    \begin{proof}
    	We have
        \[ T_n\left(\cos(\frac{\theta+2\pi i}{n})\right)=\cos(\theta+2\pi i)=\cos(\theta). \]
        For almost all $\theta$ these roots are distinct and since $T_n$ has degree $n$, it follows that they are all of the roots. When some of these roots collide, we can perturb $\theta$ and use the fact that roots are continuous functions of the polynomial coefficients to prove the statement.
    \end{proof}
    Using the above lemma for $\theta=\pi$ and $\theta=0$, we get that $\mu_i=1+\cos(\frac{\pi+2\pi i}{n})$ and $\nu_i=1+\cos(\frac{2\pi i}{n})$. This proves that $\bmu,\bnu\in\R_+^n$. Moreover we have
    \[ \frac{\nu_{\max}}{\mu_{\max}}=\frac{1+1}{1+\cos(\pi/n)}=1+\Omega(1/n^2). \]
    This finishes the proof for $k=n-1$.
    
    Now let us prove the statement for general $k$. By applying the above proof for $n=k+1$, we get $\tilde\bmu,\tilde\bnu\in \R_+^{k+1}$ such that $e_i(\tilde\bmu)=e_i(\tilde\bnu)$ for $i=1,\dots,k$ and
    \[\frac{\tilde\nu_{\max}}{\tilde\mu_{\max}}\geq 1+\Omega(1/k^2). \]
    Now construct $\bmu,\bnu$ from $\tilde\bmu,\tilde\bnu$ by 
    adding zeros to make the total count $n$. It is not hard to see, by using \autoref{cor:ekpk} that $e_i(\bmu)=e_i(\bnu)$ for $i=1,\dots,k$. Moreover $\bmu_{\max}=\tilde\bmu_{\max}$ and $\bnu_{\max}=\tilde\bnu_{\max}$. This finishes the proof.
\end{proof}
Note that the above lower bound is the same as the lower bound in \autoref{thm:negative} when $k=\Omega(n)$. However, to prove \autoref{thm:negative} and \autoref{thm:cond-negative} we need more tools. The crucial idea we use to get the stronger \autoref{thm:negative} and \autoref{thm:cond-negative} is the following observation about signed adjacency matrices for graphs of large girth.
\begin{lemma}
	\label{lem:sign-not-matter}
	Let $G=([n], E)$ be a graph and $D\in \R^{n\times n}$ an arbitrary diagonal matrix. If $\girth(G)>k$, then the top $k$ coefficients of the polynomial $\chi(x)=\det(xI-(D+A_s))$ are independent of the signing $s$. In other words, \[e_i(\blambda(D+A_{s_1}))=e_i(\blambda(D+A_{s_2})),\] for any two signings $s_1,s_2$ and $i=1,\dots,k$.
\end{lemma}


We will apply the above lemma, with $D=0$, to graphs of large girth constructed based on \autoref{conj:erdos} or \autoref{thm:large-girth}, in order to prove \autoref{thm:cond-negative} and \autoref{thm:negative} for the regime $k\leq \Theta(\log n)$. In order to prove \autoref{thm:cond-negative} for the regime $k\geq \Theta(\log n)$, we marry these constructions with Chebyshev polynomials. We will prove the above lemma at the end of this section, after proving \autoref{thm:negative} and \autoref{thm:cond-negative}.


First let us prove \autoref{thm:cond-negative}. 

\begin{proofof}{\autoref{thm:cond-negative}}
	We apply \autoref{lem:sign-not-matter} to the following graph
	construction, the proof of which we defer to \autoref{appendix:constructions}.

	\begin{claim}
		\label{claim:erdos-max-deg}
		Let $k$ be fixed and assume that \autoref{conj:erdos} is true for graphs of girth $>2k$. Then, for all sufficiently large $n$, there exist bipartite graphs $G=([n], E)$ with $\girth(G)>2k$, $\maxdeg(G)=O(n^{1/k})$, and $\avgdeg(G)=\Omega(n^{1/k})$.
	\end{claim}

	Let $G$ be the graph from the above claim. Let $s_1$ be the trivial signing that assigns $+1$ to every edge, and $s_2$ the signing guaranteed by \autoref{cor:signing}. Now let $\bnu=\blambda(A_{s_1}^2)$, i.e. the square of the eigenvalues of $A_{s_1}$, and $\bmu=\blambda(A_{s_2}^2)$, i.e. the square of the eigenvalues of $A_{s_2}$. By \autoref{lem:sign-not-matter} and \autoref{cor:ekpk}, we have
	\[ p_i(\bnu)=p_{2i}(\blambda(A_{s_1}))=p_{2i}(\blambda(A_{s_2}))=p_i(\bmu), \]
	for $i=1,\dots,k$. By another application of \autoref{cor:ekpk}, we have $e_i(\bmu)=e_i(\bnu)$ for $i=1,\dots,k$.
	
	On the other hand, by \autoref{cor:signing}, we have
	\[ \mu_{\max}=\lambda_{\max}(A_{s_2})^2\leq 4(\maxdeg(G)-1)=O(n^{1/k}), \]
	and by \autoref{fact:avgdeg}, we have
	\[ \nu_{\max}=\lambda_{\max}(A_{s_1})^2\geq \avgdeg(G)^2=\Omega(n^{2/k}). \]
	Therefore
	\[ \frac{\nu_{\max}}{\mu_{\max}}=\frac{\Omega(n^{2/k})}{O(n^{1/k})}=\Omega(n^{1/k}). \]
\end{proofof}

Now let us prove \autoref{thm:negative} for $k\leq \Theta(\log n)$.
\begin{proofof}{\autoref{thm:negative} for $k\leq \Theta(\log n)$}
	We apply \autoref{lem:sign-not-matter} to the following graph
	construction, the proof of which we defer to \autoref{appendix:constructions}.
	\begin{claim}
    	\label{claim:large-girth-const}
		Let $d$ be a prime. For all sufficiently large $n$, there exist bipartite graphs $G=([n], E)$ with $\girth(G)=\Omega(\log n)$, $\maxdeg(G)\leq d$, and $\avgdeg(G)\geq d/2$.
	\end{claim}
	We will fix $d$ to a specific prime later. Similar to the proof of \autoref{thm:cond-negative}, we let $s_1$ be the trivial signing with all $+1$s and $s_2$ be the signing guaranteed by \autoref{cor:signing}. Let $t=\Omega(\log n/k)$ be an even integer such that $tk<\girth(G)$. Such a $t$ exists when $k<c\log n$ for some constant $c$. Take $\bnu=\blambda(A_{s_1}^t)$ and $\bmu=\blambda(A_{s_2}^t)$. Then we have
	\[ \mu_{\max}=\lambda_{\max}(A_{s_2})^t\leq (2\sqrt{d-1})^t, \]
	and
	\[ \nu_{\max}=\lambda_{\max}(A_{s_1})^t\geq (d/2)^t. \]
	This means that
	\[ \frac{\nu_{\max}}{\mu_{\max}}\geq \left(\frac{d}{4\sqrt{d-1}}\right)^t\geq e^{\Omega(\log n/k)}=n^{\Omega(1/k)},
	\]
	as long as $\frac{d}{4\sqrt{d-1}}>e$, which happens for sufficiently large $d$ (such as $d=127$).
	It only remains to show that $e_i(\bmu)=e_i(\bnu)$ for $i=1,\dots,k$. For every $i\in [k]$ we have $t\cdot i< \girth(G)$, which by \autoref{lem:sign-not-matter} gives us
	\[ p_{i}(\bnu)=p_{t\cdot i}(\blambda(A_{s_1}))=p_{t\cdot i}(\blambda(A_{s_2}))=p_{i}(\bnu), \]
	and this finishes the proof because of \autoref{cor:ekpk}.
\end{proofof}

The above method unfortunately does not seem to directly extend to the regime $k\geq \Theta(\log n)$, since for large $k$, getting $\girth(G)>k$ requires many vertices of degree at most $2$.\footnote{Unless all degrees are $2$ in which case we can actually reprove \autoref{thm:weak}; we omit the details here.} Instead we use the machinery of Chebyshev polynomials to boost our graph constructions.
\begin{proofof}{\autoref{thm:negative} for $k\geq \Theta(\log n)$}
	Since \autoref{thm:weak} proves the same desired bound as \autoref{thm:negative} when $k=\Omega(n)$, we may without loss of generality assume that $n/k$ is at least a large enough constant. Let $m$ be the largest integer such that 
    \begin{equation}
    	\label{eq:m-choice}
    	c\cdot \frac{m}{\log m}\leq \frac{n}{k},
    \end{equation}
    where $c$ is a large constant that we will fix later. It is easy to see that
    \begin{equation}
    \label{eq:m-assymp}
    m=\Theta\left(\frac{n}{k}\log(\frac{n}{k})\right).
    \end{equation}
    We have already proved \autoref{thm:negative} for the small $k$ regime. Using this proof (for $n=m$ and $k=\Theta(\log m)$), we can find $\tilde\bmu,\tilde\bnu\in\R_+^m$ such that $e_i(\tilde\bmu)=e_i(\tilde\bnu)$ for $i=1,\dots,\Omega(\log m)$ and $\tilde\nu_{\max}/\tilde\mu_{\max}\geq m^{1/\Theta(\log m)}\geq 2$. Without loss of generality, by a simple scaling, we may assume that $\tilde\nu_{\max}=1$ and $\tilde\mu_{\max}\leq 1/2$.
    
    Let $\chi_{\tilde\bmu}(x),\chi_{\tilde\bnu}(x)\in\R[x]$ be the unique monic polynomials whose roots are $\tilde\bmu,\tilde\bnu$ respectively. By construction, the top $\Omega(\log m)$ coefficients of these polynomials are the same. We boost the number of equal coefficients by composing them with Chebyshev polynomials. Let $p(x):=\chi_{\tilde\bmu}(T_t(x))$ and $q(x):=\chi_{\tilde\bnu}(T_t(x))$ where $t=\lfloor n/m\rfloor$. Note that $\deg p=\deg q=tm\leq n$. We let $\bmu,\bnu$ be the roots of $p(x),q(x)$ together with some additional zeros to make their counts $n$. First, we show that the top $k$ coefficients of $p,q$ are the same. Then we show that they are real rooted, i.e., $\bmu,\bnu$ are real vectors. Finally, we lower bound $\nu_{\max}/\mu_{\max}$.
    
    Note that $p(x),q(x)$ have degree $tm$. They are not monic, but their leading terms are the same. Besides the leading terms, we claim that they have the same top $\Omega(t\log m)$ coefficients. This follows from the fact that $T_t(x)$ is a degree $t$ polynomial. When expanding either $\chi_{\tilde\bnu}(T_t(x))$ or $\chi_{\tilde\bmu}(T_t(x))$, terms of degree $\leq m-\Omega(\log m)$ in $\chi_{\tilde\bmu},\chi_{\tilde\bnu}$ produce monomials of degree at most $tm-\Omega(t\log m)$, which means that the top $\Omega(t\log m)$ coefficients are the same. This shows that the first $\Omega(t\log m)$ elementary symmetric polynomials of $\bmu,\bnu$ are the same. It follows that the first $k$ elementary symmetric polynomials to be the same, which follows from
    \[ \Omega(t\log m)=\Omega(\frac{n\log m}{m})\geq \Omega(ck)\geq k,\]
    where for the first inequality we used \eqref{eq:m-choice} and for the second inequality we assumed $c$ is large enough that it cancels the hidden constants in $\Omega$.
    
    It is not obvious if $\bmu, \bnu$ are even real. This is where we crucially use the properties of the Chebyshev polynomial $T_t(x)$. Note that the roots of $\chi_{\tilde\bmu}(x)$ and $\chi_{\tilde\bnu}(x)$, i.e. $\tilde\mu_i$'s and $\tilde\nu_i$'s are all in $[0,1]\subseteq[-1,1]$. Therefore each one of them can be written as $\cos(\theta)$ for some $\theta\in\R$. By \autoref{lem:cheb-roots} the equation
    \[ T_t(x)=\cos(\theta),\]
    has $t$ real roots (counting multiplicities), and they are simply $x=\cos(\frac{\theta+2\pi i}{t})$ for $i=0,\dots,t-1$. So for each root of $\chi_{\tilde\bmu}(x)$ we have $t$ roots of $\chi_{\tilde\bmu}(T_t(x))$, all in $[-1,1]$. This means that all of the roots of $p(x)$ are real and in $[-1, 1]$. By a similar argument, all of the roots of $q(x)$ are real and in $[-1, 1]$.
    
    The largest root of $q(x)$ is $1$, since
    \[ q(1)=\chi_{\tilde\bnu}(T_t(1))=\chi_{\tilde\bnu}(1)=\chi_{\tilde\bnu}(\nu_{\max})=0. \]
    
    On the other hand, the largest root of $p(x)$ is at most $\cos(\pi/3t)$, because for any $x\in(\cos(\pi/3t),1]$ there is $\theta\in[0,\pi/3t)$ such that $x=\cos(\theta)$ and this means
    \[ p(x)=\chi_{\tilde\bmu}(T_t(\cos(\theta)))=\chi_{\tilde\bmu}(\cos(t\theta))\neq 0, \]
    because $\cos(t\theta)>\cos(\pi/3)=1/2\geq \tilde\mu_{\max}$.
    
    By the above arguments, $\bmu,\bnu$ satisfy almost all of the desired properties, except that they could be negative. However we know that $\bmu,\bnu\in[-1,1]^n$. So using \autoref{fact:linear-transform} we can easily make them nonnegative. We simply replace $\bmu,\bnu$ by $\bmu+1,\bnu+1$. Then $\bmu,\bnu\in[0,2]^{n}$ and $e_i(\bmu)=e_i(\bnu)$ for $i=1,\dots,k$. Finally, we have
    \[ \frac{\nu_{\max}}{\mu_{\max}}\geq \frac{1+1}{1+\cos(\pi/3t)}=1+\Omega\left(\frac{1}{t^2}\right)=1+\Omega\left(\frac{m}{n}\right)^2=1+\Omega\left(\frac{\log\frac{2n}{k}}{k}\right)^2, \]
    where we used \eqref{eq:m-assymp} for the last equality.
\end{proofof}

Having finished the proofs of \autoref{thm:negative}, \autoref{thm:cond-negative} we are ready to prove \autoref{lem:sign-not-matter}.

\begin{proofof}{\autoref{lem:sign-not-matter}}
	By \autoref{cor:ekpk}, it is enough to prove that
    \begin{equation*}
    	p_k(\blambda(D+A_{s_1}))=p_k(\blambda(D+A_{s_2})),
    \end{equation*}
    for $k<\girth(G)$. This is the same as proving
    \begin{equation}
    \label{eq:power-k}
    \trace\left((D+A_{s_1})^k\right)=\trace\left((D+A_{s_2})^k\right).
    \end{equation}
    For a matrix $M\in\R^{n\times n}$ we have the following identity
    \[ \trace(M^k)=\sum_{(v_1,\dots,v_k)\in [n]^k}M_{v_1,v_2}M_{v_2,v_3}\dots M_{v_{k-1},v_k}M_{v_k, v_{1}}.\]
    For ease of notation let us identify $v_{k+1}$ with $v_1$. We apply the above formula to both sides of \eqref{eq:power-k}. The sequence $(v_1,\dots,v_k)$ can be interpreted as a sequence of vertices in the graph $G$. If for any $i$, $v_i\neq v_{i+1}$ and $\{v_i,v_{i+1}\}\notin E$, then the term inside the sum vanishes. Therefore we can restrict the sum to those terms $(v_1,\dots,v_k)$ where for each $i\in[k]$, either $v_i=v_{i+1}$ or $v_i$ and $v_{i+1}$ are connected in $G$. To borrow and abuse some notation from Markov chains, let us call such a sequence a lazy closed walk of length $k$. In order to prove \eqref{eq:power-k} it is enough to prove that for any such lazy closed walk we get the same term for both $s_1$ and $s_2$.
    
    Let $(v_1,\dots,v_k)$ be one such lazy closed walk. Consider a particle that at time $i$ resides at $v_i$. In each step, the particle either does not move or moves to a neighboring vertex, and at time $k+1$ it returns to its starting position. For each step that the particle does not move we get one of the entries of $D$, corresponding to the current vertex, as a factor. This is clearly independent of the signing. When the particle moves however, we get the sign of the edge over which it moved as a factor. We will show that the particle must cross each edge an even number of times. Therefore the signs for each edge cancel each other and we get the same result for $s_1,s_2$.
    
    Consider the induced subgraph on $v_1,\dots, v_k$. Because $k<\girth(G)$, this subgraph has no cycles; therefore it must be a tree. A (lazy) closed walk crosses any cut in a graph an even number of times. Each edge in this tree constitutes a cut. Therefore the lazy closed walk $(v_1,\dots,v_k)$ must cross each edge in the tree an even number of times.
\end{proofof}

\subsection{Lower Bound Given Approximate Coefficients}
The lower bounds proved in the previous section show that knowing a small number of the coefficients of a polynomial {\em exactly} is insufficient to obtain a good estimate of its largest root. 

In this section we generalize the construction of \autoref{thm:negative} to provide a
satisfying lower bound to \autoref{prob:approx}.
\begin{proposition} 
\label{prop:approxlb}
For every integer $n>1$ and $1<k<n$ there are 
degree $n$ polynomials $r(x),s(x)$ such that 
\begin{enumerate}
	\item All of the coefficients of $r$ and $s$ except for the $2k^{th}$ are exactly
	equal, and the $2k^{th}$ coefficients are within a multiplicative factor
	of $1+\frac{4}{2^{2k}}.$
\item The largest root of $r$ is at least $1+\Omega(1/k^2)$ of the largest root of $s$.
\item $r$ and $s$ have a common interlacing.
\item $r$ and $s$ are characteristic polynomials of graph Laplacians. Further,
	these Laplacians correspond to $2-$lifts of a common base graph.
\end{enumerate}
\end{proposition}
\begin{proof}
	Let 
	$$r(x):=2T_k^2(3/2-x)\qquad and\qquad s(x):=T_{2k}(3/2-x).$$
	Since
	$$T_{2k}=2T_k^2-1,$$
	these polynomials differ only in their constant terms. Moreover, we have
	$$r(0)=2T_k^2(3/2)\ge (2^{k-1})^2\qquad and\qquad s(0)\ge 2^{2k-1}/2$$
	by \autoref{fact:cheb-large-x}. Thus, they agree on the first $2k-1$
	coefficients, and differ by a multiplicative factor of at most
	$$\frac{2^{2k-2}+1}{2^{2k-2}}=1+4/2^{2k}$$ in the $2k^{th}$ coefficient,
	establishing  (1).

	To see (2), observe that $r(x)$ has largest root $3/2+\cos(2\pi/k)$
	whereas $s(x)$ has largest root $3/2+\cos(2\pi/2k)$. Since the
	difference of these numbers is
	$$\cos(2\pi/k)-\cos(2\pi/2k) =
	\frac{4\pi^2}{2}(\frac{1}{k^2}-\frac{1}{4k^2}+o(1/k^4)),$$
	we conclude that their ratio is at least
	$1+\Omega(1/k^2)$. 

	To see (3), observe that the roots of $r(x)$ are 
	$$ 3/2-\cos(2\pi j/k)\quad j=0,\ldots,k-1\quad\textrm{with multiplicity
	$2$}$$
	and the roots of $s(x)$ are
	$$ 3/2-\cos(2\pi j/2k)\quad j=0,\ldots,2k-1\quad\textrm{with
	multiplicity $1$},$$
	whence $r$ and $s$ have a common interlacing according to
	\autoref{defn:interlacing}.

	For (4) we first apply \autoref{fact:godsil} to interpret both $r$ and
	$s$ as characteristic polynomials of cycles.
	Let $C_k$ denote a cycle of length $k$ and let $C_k\cup C_k$ denote a
	union of two such cycles. Then we have:
	$$ \det(2xI-A_{C_k\cup C_k}) = \det(2xI-A_{C_k})^2 = 4T_k(x)^2 =
	2r(3/2-x),$$
	and
	$$\det(2xI-A_{C_{2k}}) = 2T_{2k}(x) = 2s(3/2-x),$$
	whence
	$$r(x) = \frac12 \det(3I-A_{C_k\cup
	C_k}-2xI)=\frac{(-2)^k}{2}\det(xI-((3/2)I-(1/2)A_{C_k\cup C_k})$$
	and
	$$s(x)=\frac12 \det(3I-A_{C_{2k}}-2xI) =
	\frac{(-2)^k}{2}\det(xI-((3/2)I-(1/2)C_{2k}),$$
	which are characteristic polynomials of graph Laplacians of weighted
	graphs with self loops. Note that both graphs are $2-$covers of $C_k$. Since multiplying by constants does not change
	any of the properties we are interested in, we can ignore them.
	
	Considering $\tilde{r}(x)=x^{n-k}r(x)$ and $\tilde{s}(x)=x^{n-k}s(x)$
	yields examples of the desired dimension $n$; note that multiplying by
	$x^{n-k}$ simply corresponds to adding isolated vertices to the
	corresponding graphs.
\end{proof}
\section{Applications to Interlacing Families}\label{section:apps}
In this section we use \autoref{thm:positive} to give an $2^{\tilde O(\sqrt[3]{m})}$ time algorithm for rounding an interlacing family of depth $m$. Let us start by defining an interlacing family. 

\begin{definition}[Interlacing]\label{defn:interlacing}
We say that a real rooted polynomial $g(x) = \alpha_0 \prod_{i=1}^{n-1} (x-\alpha_i)$ interlaces a real rooted polynomial $f(x) = \beta_0 \prod_{i=1}^n (x-\beta_i)$ if
$$\beta_1 \leq \alpha_1 \leq \beta_2 \leq \alpha_2 \leq \dots \leq \alpha_{n-1} \leq \beta_n. $$
\end{definition}
We say that polynomials $f_1,\dots,f_k$ have a common interlacing if there is a polynomial $g$ such that $g$ interlaces all $f_i$.
The following key lemma is proved in \cite{MSS12}.
\begin{lemma}
\label{lem:interlacingroot}
Let $f_1,\dots,f_k$ be polynomials of the same degree that are real rooted and have positive leading coefficients. 
If $f_1,\dots,f_k$ have a common interlacing, then there is an $i$
such that
$$ \lambda_{\max}(f_i) \leq \lambda_{\max}(f_1+\dots+f_k).$$
\end{lemma}

\begin{definition}[Interlacing Family]
\label{def:interlacingfamily}
Let $S_1,\dots,S_m$ be finite sets.
Let $\cF\subseteq S_1\times S_2\times \dots \times S_m$ be nonempty. For any $s_1,s_2,\dots,s_m\in\cF$,  let $f_{s_1,\dots,s_m}(x)$ be a real rooted polynomial of degree $n$ with a positive leading coefficient.
For $s_1,\dots,s_k\in  S_1\times \dots\times S_k$ with $k<m$, let
$$ \cF_{s_1,\dots,s_k}:=\{t_1,\dots,t_m\in\cF: s_i=t_i, \forall 1\leq i\leq k\}.$$
Note that $\cF=\cF_{\emptyset}$.
Define
$$f_{s_1,\dots,s_k} = \sum_{t_1,\dots,t_m \in\cF_{s_1,\ldots,s_k}} f_{t_1,\dots,t_m},$$
and
$$ f_{\emptyset} = \sum_{t_1,\dots,t_m\in\cF} f_{t_1,\dots,t_m}.$$
We say polynomials $\{f_{s_1,\dots,s_m}\}_{s_1,\dots,s_m\in\cF}$ form an {\em interlacing family} if for all $0\leq k<m$ and all
$s_1,\dots,s_k\in S_1\times \dots\times S_k$ the following holds: The polynomials 
$f_{s_1,\dots,s_k,t_i}$ which are not identically zero 
have a common interlacing. 
\end{definition}
In the above definition we say $m$ is the depth of the interlacing families.
It follows by repeated applications of \autoref{lem:interlacingroot} that for any interlacing family $\{f_{s_1,\dots,s_m}\}_{s_1,\dots,s_m\in\cF}$, there is a polynomial $f_{s_1,\dots,s_m}$ such that the largest root of $f_{s_1,\dots,s_m}$ is at most the largest root of $f_{\emptyset}$ {\cite[Thm 3.4]{MSS13}}.
For an $\alpha>1$, an $\alpha$-approximation {\em rounding} algorithm for an interlacing family $\{f_{s_1,\dots,s_m}\}_{s_1,\dots,s_m\in\cF}$ is an algorithm that returns a polynomial $f_{s_1,\dots,s_m}$ such that
$$ \lambda_{\max}(f_{s_1,\dots,s_m}) \leq \alpha\lambda_{\max}(f_{\emptyset}).$$
Next, we design such a rounding algorithm, given an oracle that computes the first $k$ coefficients of the polynomials in an interlacing family.
\begin{theorem}\label{thm:rounding}
Let $S_1,\dots,S_m$ be finite sets and let $\{f_{s_1,\dots,s_m}\}_{s_1,\dots,s_m\in\cF}$ be an interlacing family of degree $n$ polynomials. Suppose that  we are given an oracle that for any $1\leq k\leq n$ and $s_1,\dots,s_\ell\in S_1,\dots,S_\ell$ with $\ell <m$ returns the top $k$ coefficients of $f_{s_1,\dots,s_m}$ in time $T(k)$. Then, there is an algorithm that for any $\eps>0$ returns a polynomial $f_{s_1,\dots,s_m}$ such that the largest root of $f_{s_1,\dots,s_m}$ is at most $1+\eps$ times the largest root of $f_{\emptyset}$, in time
$T(O(\log(n)m^{1/3}/\sqrt{\eps})) \max\{|S_i|\}^{O(m^{1/3})}\poly(n)$.
\end{theorem}
\begin{proof}
Let $M=m^{1/3}$ and $k\asymp \frac1{\sqrt{\eps}} \log(n)M$. Then, by \autoref{thm:positive}, for any polynomial $f_{s_1,\dots,s_\ell}$ we can find a $1+\frac{\eps}{2M^2}$ approximation of the largest root of $f_{s_1,\dots,s_m}$ in time $T(k)\poly(n)$. We round the interlacing family in $M^2$ many steps and in each step we round $M$ of the coordinates. We make sure that each step only incurs a (multiplicative) approximation of $1+\frac{\eps}{2M^2}$ so that the cumulative approximation error is no more than
$$ (1+\frac{\eps}{2M^2})^{M^2}\leq 1+\eps$$
as desired.

Let us describe the algorithm inductively. Suppose we have selected $s_1,\dots,s_\ell$ for some $0\leq \ell< m$. We brute force over all polynomials $f_{s_1,\dots,s_{\ell},t_{\ell+1},\dots,t_{\ell+M}}$ which are not identically zero for all $t_{\ell+1},\dots,t_{\ell+M} \in S_{\ell+1},\dots,S_{\ell+M}$. Note that there are at most $(\max_i |S_i|)^{M}$ many such polynomials. For any polynomial $f_{s_1,\dots,s_\ell,t_{\ell+1},\dots,t_{\ell+M}}$ (which is not identically zero) we compute  $\mu^*_{s_1,\dots,s_\ell,t_{\ell+1},\dots,t_{\ell+M}}$, a $1+\frac{\eps}{2M^2}$ approximation of its largest root using its top $k$ coefficients.
We let
$$ s_{\ell+1},\dots,s_{\ell+M}=\ \ \smashoperator{ \argmin_{t_{\ell+1},\dots,t_{\ell+M}}}\ \  \mu^*_{s_1,\dots,s_\ell,t_{\ell+1},\dots,t_{\ell+M}}.$$

It follows that the algorithm runs in time $T(k) \poly(n) (\max_i |S_i|)^{O(M)}$. Because we have an interlacing family, there is a polynomial $f_{s_1,\dots,s_\ell,t_{\ell+1},\dots,t_{\ell+M}}$ whose largest root is at most the largest root of $f_{s_1,\dots,s_\ell}$, in each step of the algorithm. Therefore,
$$ \lambda_{\max}(f_{s_1,\dots,s_{\ell+M}}) \leq (1+\frac{\eps}{2M^2})\lambda_{\max}(f_{s_1,\dots,s_\ell}).$$
Therefore, by induction,
$$ \lambda_{\max}(f_{s_1,\dots,s_m}) \leq (1+\eps) \lambda_{\max}(f_{\emptyset})$$
as desired.
\end{proof}

We remark that without the use of Chebyshev polynomials and
\autoref{thm:positive}, one can obtain the somewhat worse running time of $2^{\tilde
O(\sqrt{m})}$ by applying the same trick of rounding the vectors in groups
rather than one at a time.

To use the above theorem, we need to construct the aforementioned oracle for each application of the interlacing families.
Next, we construct such an oracle for several examples.
\subsection{Oracle for Kadison-Singer}
We start with interlacing families corresponding to the Weaver's problem which is an equivalent formulation of the Kadison-Singer problem.
Marcus, Spielman, and Srivastava proved the following theorem.
\begin{theorem}[\cite{MSS13}]
Given vectors $\bv_1,\dots,\bv_m$ in isotropic position,
$$ \sum_{i=1}^m \bv_i\bv_i^T = I,$$
such that $\max_i \norm{\bv_i}^2 \leq \delta$,
there is a partitioning $S_1,S_2$ of $[m]$ such that for $j\in\{1,2\}$,
$$ \norm{\sum_{i\in S_j}\bv_i\bv_i^T} \leq 1/2+O(\sqrt{\delta}).$$
\end{theorem}
We give an algorithm that finds  the above partitioning and runs in time $2^{m^{1/3}/\delta}$. It follows from the proof of \cite{MSS13}
that it is enough to design a $(1+\delta)$-approximation rounding algorithm for the following interlacing family.

Let $\br_1,\dots,\br_m\in\R^n$ be independent random vectors where for each $i$, $S_i$ is the support of $\br_i$. For any $\br_1,\dots,\br_m\in S_1\times\dots S_m$ let
$$ f_{\br_1,\dots,\br_m}(x)=\det(xI-\sum_{i=1}^m \br_i\br_i^T).$$
Marcus et al.~\cite{MSS13} showed that $\{f_{\br_1,\dots,\br_m}\}_{\br_1,\dots,\br_m\in S_1\times \dots\times S_m}$ is an interlacing family. Next we design an algorithm that returns the first $k$ coefficients of any polynomial in this family in time $(m\cdot \max_i |S_i|)^k$.

\begin{theorem}\label{thm:KSoracle}
Given independent random vectors $\br_1,\dots,\br_m\in \R^d$ with support $S_1,\dots,S_m$. There is an algorithm that for any $\br_1,\dots,\br_\ell\in S_1\times\dots\times S_\ell$ with $1\leq \ell \leq m$ and $1\leq k\leq n$ returns the top $k$ coefficients of $f_{\br_1,\dots,\br_\ell}$ in time $(m\cdot\max_i |S_i|)^k \poly(n)$. 
\end{theorem}
\begin{proof}
Fix $\br_1,\dots,\br_\ell\in S_1\times \dots \times S_\ell$ for some $1\leq \ell \leq m$.
It is sufficient to show that for any $0\leq k\leq n$, we can compute the coefficient of $x^{n-k}$ of $f_{\br_1,\dots,\br_\ell}$ in time $(m\cdot \max_i |S_i|)^k \poly(n)$.
First, observe that
$$ f_{\br_1,\dots,\br_\ell}=\EE{\br_{\ell+1},\dots,\br_m}{\det\left(xI-\sum_{i=1}^m \br_i\br_i^T\right)}.$$
So, by \autoref{prop:sigmak}, the coefficient of $x^{n-k}$ in the above is equal to
$$ (-1)^k\sum_{T\subseteq \binom{[m]}{k}} \EE{\br_{\ell+1},\dots,\br_m}{\sigma_k\left(\sum_{i\in T} \br_i\br_i^T\right)}.$$
Note that there are at most $m^k$ terms in the above summation. For any $T\subseteq \binom{[m]}{k}$, we can exactly compute 
$\EE{\br_{\ell+1},\dots,\br_m}{\sigma_k(\sum_{i\in T}\br_i\br_i^T)}$ in time 
$(\max_i |S_i|)^k \poly(n)$. All we need to do is brute force over all vectors in the domain of $\{S_i\}_{i>\ell, i\in T}$ and average out $\sigma_k(.)$ of the corresponding sums of rank 1 matrices.
\end{proof}
It follows by \autoref{thm:rounding} and \autoref{thm:KSoracle} that for any given set of vectors $\bv_1,\dots,\bv_m\in \R^n$ in isotropic position of squared norm at most $\delta$ we can find a two partitioning $S_1,S_2$ such that $\norm{\sum_{i\in S_j} \bv_i\bv_i^T}\leq 1/2+O(\sqrt{\delta})$ in time $ n^{O(m^{1/3}\delta^{-1/4})} $.

\subsection{Oracle for ATSP}
Next, we construct an oracle for interlacing families related to the asymmetric traveling salesman problem.
We say a multivariate polynomial $p\in\R[z_1,\dots,z_m]$ is real stable if it has no roots in the upper-half complex plane, i.e., $p(\bz)\neq 0$ whenever $\text{Im}(z_i)>0$ for all $i$.
The {\em generating polynomial} of $\mu$ is defined as
$$ g_\mu(z_1,\dots,z_m)=\sum_{S\subseteq [m]} \mu(S)\prod_{i\in S} z_i.$$
We say $\mu$ is a {\em strongly Rayleigh}  probability distribution if $g_{\mu}(\bz)$ is real stable. We say $\mu$ is homogeneous if all sets in the support of $\mu$ have the same size. 
The following theorem is proved by the first and the second authors~\cite{AO15}.
\begin{theorem}[\cite{AO15}]
Let $\mu$ be a homogeneous strongly Rayleigh probability distribution on subsets of $[m]$ such that for each $i$, $\P{i}<\eps_1$. Let $\bv_1,\dots,\bv_m\in\R^n$ be in isotropic position such that for each $i$, $\norm{\bv_i}^2\leq \eps_2$. Then, there is a set $S$ in the support of $\mu$ such that
$$ \norm{\sum_{i\in S}\bv_i\bv_i^T} \leq O(\eps_1+\eps_2).$$
\end{theorem}
We give an algorithm that finds such a set $S$ as promised in the above theorem assuming that we have an oracle that for any $\bz\in \R^m$ returns $g_{\mu}(\bz)$.
It follows from the proof of \cite{AO15} that it is enough to design a $1+O(\eps_1+\eps_2)$-approximation rounding algorithm for the following interlacing family.

For any $1\leq i\leq m$, let $S_i=\{\bzero,\bv_i\}$. For any $S$ in the support of $\mu$, let $\br_i=\bv_i$ if $i\in S$ and $\br_i=\bzero$ otherwise and we add $\br_1,\dots,\br_m$ to $\cF$. Then, define
$$ f_{\br_1,\dots,\br_m}=\mu(S)\cdot \det\left(xI-\sum_{i=1}^m \br_i\br_i^T\right).$$
It follows from \cite{AO15} that $\{f_{\br_1,\dots,\br_m}\}_{\br_1,\dots,\br_m\in \cF}$ is an interlacing family.

Next, we design an algorithm that returns the top $k$ coefficients of any polynomial in this family in time $m^k\poly(n)$.
\begin{theorem}\label{thm:SRoracle}
Given a strongly Rayleigh distribution $\mu$ on subsets of $[m]$ and a set of vectors $\bv_1,\dots,\bv_m\in\R^n$, suppose that we are given an oracle that for any $\bz\in\R^m$ returns $g_\mu(\bz)$. There is an algorithm that for any $\br_1,\dots,\br_\ell\in S_1\times\dots\times S_\ell$ with $1\leq \ell\leq m$ and $1\leq k\leq n$ returns the top $k$ coefficients of $f_{\br_1,\dots,\br_\ell}$ in time $m^k\poly(n)$.
\end{theorem}
\begin{proof}
Fix $\br_1,\dots,\br_\ell\in S_1\times \dots\times S_\ell$ for some $1\leq \ell\leq m$. 
First, note that if there is no such element in $\cF$, then $g_\mu(1,\dots,1,z_{\ell+1},\dots,z_m)=0$ and there is nothing to prove. So assume for some $\br_{\ell+1},\dots,\br_m$, $\br_1,\dots,\br_m\in \cF$.

It is sufficient to show that for any $0\leq k\leq n$, we can compute the coefficient of $x^{n-k}$ of $f_{\br_1,\dots,\br_\ell}$ in time $m^k\poly(n)$. 
Firstly, observe that  since we are only summing up the characteristic polynomials that are consistent with $\br_1,\dots,\br_\ell$, we can work with the conditional distribution
$$ \tilde{\mu}=\{\mu \,|\, i \text{ if } 1\leq i\leq \ell \text{ and } \br_i=\bv_i, \overline{i}\text{ if } 1\leq i\leq \ell, \br_i=\bzero\}.$$
Note that since we can efficiently compute $g_\mu(z_1,\dots,z_m)$, we can also compute $g_{\tilde{\mu}}(z_{\ell+1},\dots,z_m)$. 
For any $i$ that is conditioned to be in, we need to differentiate with respect to $z_i$ and for any $i$ that is conditioned to be out we let $z_i=0$. Also, note that instead of differentiating we can let $z_i=M$ for a very large number $M$, and then divide the resulting polynomial by $M$. We note that when $\mu$ is a determinantal distribution, which is the case in applications to the asymmetric traveling salesman problem, this differentiation can be computed exactly and efficiently; in other cases, $M$ can be taken to be exponentially large, as we can tolerate an exponentially small error.

Now, we can write
$$ f_{\br_1,\dots,\br_\ell}= \EE{T\sim\tilde{\mu}}{\det\left(xI-\sum_{i\in T}\bv_i\bv_i^T\right)}.$$
So, by \autoref{prop:sigmak}, the coefficient of $x^{n-k}$ in the above is equal to
$$ (-1)^k \sum_{T\in \binom{[m]}{k}} \PP{\tilde{\mu}}{T}\cdot \sigma_k\left(\sum_{i\in T} \bv_i\bv_i^T\right). $$
To compute the above quantity it is enough to brute force over all sets $T\in \binom{[m]}{k}$. For any such $T$ we can compute $\sigma_{k}(\sum_{i\in T} \bv_i\bv_i^T)$ in time $\poly(n)$. In addition, we can efficiently compute $\PP{\tilde{\mu}}{T}$ using our oracle. It is enough to differentiate with respect to any $i\in T$,
$$ \prod_{i\in T: i>\ell} \frac{\partial}{\partial z_i} g_{\tilde{\mu}}(z_{\ell+1},\dots,z_m)|_{z_{\ell+1}=\dots=z_m=1}.$$
Therefore, the algorithm runs in time $m^k\poly(n)$. 
\end{proof}
It follows from \autoref{thm:rounding} and \autoref{thm:SRoracle} that for any homogeneous strongly Rayleigh distribution $\mu$ with marginal probabilities $\eps_1$ with an oracle that computes $g_\mu(\bz)$ for any $\bz\in\R^n$, and for any vectors $\bv_1,\dots,\bv_m$ in isotropic position with squared norm at most $\eps_2$, we can find a set $S$ in the support of $\mu$ such that
$\norm{\sum_{i\in S}\bv_i\bv_i^T}\leq O(\eps_1+\eps_2)$ in time $n^{O(m^{1/3}(\eps_1+\eps_2)^{1/2})}$.

This is enough to get a $\polyloglog(m)$ approximation algorithm for asymmetric traveling salesman problem on a graph with $m$ vertices that runs in time $2^{\tilde O(m^{1/3})}$ \cite{AO14}.

\bibliographystyle{alpha}
\bibliography{main-1.bbl}
\appendix
\section{Appendix A: Proofs of Preliminary Facts}\label{appendix:prelims}
\begin{proof}[Proof of \autoref{fact:linear-transform}]
	If $a=0$, the conclusion is trivial. Otherwise let $\l(x)=(x-b)/a$. Then $a\bmu+b$ is the vector of the roots of $\chi_\bmu(\l(x))$ and similarly $a\bnu+b$ is the vector of the roots of $\chi_\bnu(\l(x))$. It is easy to see that if $\chi_{\bmu}$ and $\chi_{\bnu}$ have the same top $k$ coefficients, then after expansion $\chi_{\bmu}(\l(x))$ and $\chi_{\bnu}(\l(x))$ have the same top $k$ coefficients as well, since $\l(x)$ is linear.
\end{proof}
\begin{proof}[Proof of \autoref{cor:compute}]
	For $p(x)=x^k$, the statement of the corollary becomes the same as \autoref{thm:newton}. For any other polynomial, we can write
	\[ p(x)=a_0+a_1x+\dots+a_kx^k. \]
	Then we have
	\[ \sum_{i=1}^n p(\mu_i)=\sum_{j=1}^k a_j\left(\sum_{i=1}^n\mu_i^j\right) =\sum_{j=1}^k a_jp_j(\bmu). \]
	Now because of \autoref{thm:newton}, it follows that
	\[ \sum_{i=1}^n p(\mu_i)=\sum_{j=1}^k a_j q_j(e_1(\bmu),\dots,e_k(\bmu)). \]
	The above is a polynomial of $e_1(\bmu),\dots,e_k(\bmu)$, and it can be computed in time $\poly(k)$, since each $q_j$ can be computed in time $\poly(k)$.
\end{proof}

\begin{proof}[Proof of \autoref{fact:cheb-large-x}]
	For $x\geq 0$, we have $T_k(1+x)=\cosh(k\cosh^{-1}(1+x))$; here $\cosh^{-1}$ is the inverse of $\cosh$ when looked at as a function from $[0,\infty)$ to $[1,\infty)$. Both $\cosh$ and $\cosh^{-1}$ are monotonically increasing over the appropriate ranges. Therefore $T_k(1+x)$ is monotonically increasing for $x\geq 0$.
	
	For $x\geq 0$, we have
	\begin{align*}
		 T_k(1+x)&=\cosh(k \cosh^{-1}(1+x))\geq \frac{\exp(k \cosh^{-1}(1+x))}{2} \\
		 &=\frac{\exp(\cosh^{-1}(1+x))^k}{2}\geq \frac{(1+\sqrt{2x})^k}{2}.
	\end{align*}
	In the first inequality we used the fact that $\cosh(x)\geq e^x/2$, and in the last inequality we used the fact that $\exp(\cosh^{-1}(1+x))\geq 1+\sqrt{2x}$.
	
\end{proof}

\begin{proof}[Proof of \autoref{fact:avgdeg}]
	The maximum eigenvalue of $A$ is characterized by the so called Rayleigh quotient
	\[ \lambda_{\max}(A)=\max_{\bx\in \R^n\setminus \{0\}} \frac{\bx^\intercal A\bx}{\bx^\intercal \bx},
	\]
	For $\bx=\bone$, the all $1$s vector, we have
	\[
		\lambda_{\max}(A)\geq \frac{\bone^\intercal A\bone}{\bone^\intercal\bone}=\frac{2|E|}{n}=\avgdeg(G).
	\]
\end{proof}
\section{Appendix B: Constructions of Graphs with Large Girth}
\label{appendix:constructions}
In this section we prove \autoref{claim:erdos-max-deg} and \autoref{claim:large-girth-const}.

\begin{proofof}{\autoref{claim:erdos-max-deg}}
	Assume that $k$ is fixed and \autoref{conj:erdos} is true for graphs of girth $>2k$. The graphs promised by \autoref{conj:erdos} already have average degree $\Omega(n^{1/k})$. We just need to make them bipartite and make sure that their maximum degree is at most $O(n^{1/k})$.
    
    Making the graph bipartite is easy. Given a graph $G$, the following procedure makes it bipartite: Replace each vertex $v$ by $v_1, v_2$. Replace each edge $\{u, v\}$ with two edges $\{v_1, u_2\}$ and $\{u_1, v_2\}$. This procedure doubles the number of edges and the number of vertices, and it is easy to see that it does not decrease the girth.
    
    We can trim the maximum degree by the following procedure: If there is a vertex $v$ where $\deg v>n^{1/k}$, we introduce a new vertex $v'$ and take an arbitrary subset of $n^{1/k}$ edges incident to $v$ and change their endpoint from $v$ to $v'$. By repeating this procedure the graph will eventually have maximum degree bounded by $n^{1/k}$. This procedure increases the number of vertices, but the number of new vertices is easily bounded by $2|E|/n^{1/k}$ because each new vertex has degree $n^{1/k}$ until the end and the number of edges $|E|$ does not change.
    
    Both of the above procedures change the number of vertices. To get graphs of arbitrary size $n$, we can use the following trick: We start with a graph $G_0$ promised by \autoref{conj:erdos} on $n/c$ vertices, where $c$ is a large constant. We make sure that it has no more than $(n/c)^{1+1/k}$ edges by removing edges if necessary (this can only increase girth). Now we make a bipartite graph $G_1$ from $G_0$ by the aforementioned procedure. From $G_1$ we make $G_2$ using the second procedure which makes sure $\maxdeg(G_2)=O(n^{1/k})$. The number of new vertices added by this step is $O(n/c)$. So at the end the total number of vertices in $G_2$ is $O(n/c)$ and if we make $c$ large enough we can make sure it is less than $n$. We can add isolated vertices to $G_2$ until the total number of vertices becomes $n$ and we get $G$. This procedure changes $\avgdeg$ by at most a constant factor and does not change $\maxdeg$, so both are $\Theta(n^{1/k})$ at the end.
\end{proofof}

\begin{proofof}{\autoref{claim:large-girth-const}}
	We build our graphs using \autoref{thm:large-girth}. Let $t\geq 3$ be the largest odd number such that $2d^t\leq n$. If $n$ is large enough, such a $t$ exists, and we have $t=\Omega(\log n)$. By \autoref{thm:large-girth}, there exists a $d$-regular bipartite graph $H$ with $2d^t$ vertices and $\girth(H)\geq t+5=\Omega(\log n)$. To construct $G$, put as many copies of $H$ as possible side by side making sure the number of vertices does not grow more than $n$. At the end add some isolated vertices to make the total number of vertices $n$ and get $G$. It is easy to see that number of isolated vertices added at the end is at most $n/2$. These vertices have degree $0$ and the rest have degree $d$. Therefore $\avgdeg(G)\geq d/2$ and $\maxdeg(G)=d$.
\end{proofof}

\end{document}